\documentclass[11pt]{article}
\usepackage[utf8]{inputenc}

\usepackage[maxbibnames=99]{biblatex}
\usepackage{amsmath}
\usepackage{amssymb}
\usepackage{amsthm}
\usepackage{extarrows}
\usepackage{tikz}
\usepackage{graphicx}
\usepackage{float}
\usepackage{algorithm}
\usepackage[noend]{algpseudocode}
\usepackage{tabularx}
\usepackage[margin=1in]{geometry}
\usepackage{arydshln}
\usepackage{multirow}
\usepackage{wrapfig}
\usepackage{url}
\usepackage{hyperref}
\usepackage{xcolor}
\usepackage{lineno}
\newcommand*\patchAmsMathEnvironmentForLineno[1]{%
  \expandafter\let\csname old#1\expandafter\endcsname\csname #1\endcsname
  \expandafter\let\csname oldend#1\expandafter\endcsname\csname end#1\endcsname
  \renewenvironment{#1}%
     {\linenomath\csname old#1\endcsname}%
     {\csname oldend#1\endcsname\endlinenomath}}%
\newcommand*\patchBothAmsMathEnvironmentsForLineno[1]{%
  \patchAmsMathEnvironmentForLineno{#1}%
  \patchAmsMathEnvironmentForLineno{#1*}}%
\AtBeginDocument{%
\patchBothAmsMathEnvironmentsForLineno{equation}%
\patchBothAmsMathEnvironmentsForLineno{align}%
\patchBothAmsMathEnvironmentsForLineno{flalign}%
\patchBothAmsMathEnvironmentsForLineno{alignat}%
\patchBothAmsMathEnvironmentsForLineno{gather}%
\patchBothAmsMathEnvironmentsForLineno{multline}%
}

\makeatletter
\newenvironment{process}[1][htb]{%
    \renewcommand{\ALG@name}{Process}
   \begin{algorithm}[#1]%
  }{\end{algorithm}}
\makeatother

\makeatletter
\DeclareRobustCommand{\rvdots}{%
  \vbox{
    \baselineskip4\p@\lineskiplimit\z@
    \kern-\p@
    \hbox{.}\hbox{.}\hbox{.}
  }}
\makeatother

\usetikzlibrary{positioning,calc,shapes,arrows}
\usetikzlibrary{backgrounds}
\usetikzlibrary{matrix,shadows,arrows} 
\usetikzlibrary{decorations.pathreplacing,calligraphy}

\def\etal.{et\penalty50\ al.}
\theoremstyle{plain}
\newtheorem{theorem}{Theorem}[section]
\newtheorem{lemma}[theorem]{Lemma}

\newtheorem{corollary}[theorem]{Corollary}

\theoremstyle{definition}
\newtheorem{definition}{Definition}[section]

\theoremstyle{remark}

\theoremstyle{plain}
\newtheorem*{theorem*}{Theorem}

\DeclareMathOperator*{\E}{\mathbb{E}}

\title{Random-Order Interval Selection}
\author{Allan Borodin \\ \textsf{bor@cs.toronto.edu} \and Christodoulos Karavasilis \\ \textsf{ckar@cs.toronto.edu}}
\date{}

\begin{document}

\maketitle
\begin{abstract}
In the problem of online unweighted interval selection, the objective is to maximize the number of non-conflicting intervals accepted by the algorithm. In the conventional online model of irrevocable decisions, there is an $\Omega(n)$ lower bound on the competitive ratio, even for randomized algorithms Bachmann et al. \cite{bachmann2013online}. In a line of work that allows for revocable acceptances, Faigle and Nawijn \cite{faigle1995note} gave a greedy $1$-competitive (i.e. optimal) algorithm in the real-time model, where intervals arrive in order of non-decreasing starting times. The natural extension of their algorithm in the adversarial (any-order) model is $2k$-competitive as shown in Borodin and Karavasilis \cite{borodin2023any}, when there are at most $k$ different interval lengths, and that is optimal for all deterministic, and memoryless randomized algorithms. We study the interval selection problem in the random-order model, where the adversary chooses the instance, but the online sequence is a uniformly random permutation of the items. We consider the same algorithm that is optimal in the cases of the real-time and any-order models, and give an upper bound of $2.5$ on the competitive ratio under random-order arrivals. We also contrast this to the best known $6$-competitive randomized algorithm in the adversarial input model by Emek et al. \cite{emek2016space}.

Our interest in the random order model leads us to 
 initiate a study of utilizing random-order arrivals to extract random bits with the goal of derandomizing algorithms. Besides producing simple algorithms, simulating random bits through random arrivals enhances our understanding of the comparative strength of randomized algorithms (with adversarial input sequence) and deterministic algorithms in the random order model. We consider three $1$-bit random extraction processes. Our third extraction process returns a bit with a worst-case bias of $2 - \sqrt{2} \approx 0.585$ and operates under the mild assumption that there exist at least two distinct items in the input.
We motivate the applicability of this process by using it to simulate a number of barely random algorithms for weighted interval selection (single-length arbitrary weights, as well as C-benevolent instances), the general knapsack problem, string guessing, minimum makespan scheduling, and job throughput scheduling.
\end{abstract}

\section{Introduction}
In the problem of interval scheduling on a single machine, there is a set of intervals on the real line, each with a fixed starting time and end time, and we must choose a subset of non-conflicting intervals. In the unweighted setting, the goal is to maximize the cardinality of the subset. In terms of the objective function, this is equivalent to finding a maximum independent set of an interval graph. In weighted variations, each interval is associated with a weight, and we aim to maximize the total weight of the solution. In the traditional online version of the problem, intervals arrive one at a time, and the algorithm must either permanently accept an interval, or forever discard it. Results in this worst case adversarial setting are quite negative. Even in the real-time model where intervals arrive in order of increasing starting times, there is an $\Omega(n)$ lower bound for randomized algorithms \cite{bachmann2013online}. Following existing work on interval scheduling, we consider a model where any new interval can be accepted, displacing any conflicting intervals currently in the solution. Similar to intervals that are rejected upon arrival, displaced intervals  can never be taken again. 
 Different types of revocable decisions are problem-specific, and appear under various names, such as \textit{preemption}, \textit{replacement}, \textit{free disposal}, and \textit{recourse}. Our use of revocable means revocable acceptances (also referred to as late rejections) and we view it as the simplest and easiest (to implement) form of preemption and recourse.  Other examples of problems that have been studied in the revocable acceptances model include the knapsack problem Iwama et al. \cite{iwama2002removable}, submodular maximization 
Buchbinder at al. \cite{buchbinder2019online}, weighted matching 
Feldman et al. \cite{feldman2009online}, maximum coverage Rawitz and Ros{\'e}n  \cite{rawitz2021online}, and other graph problems Boyar et al. \cite{boyar2022relaxing}. It is worth noting that algorithms in these models are particularly relevant in the context of size of budget constraints and are also relevant when online algorithms are used to construct offline solutions. \\\\

In the adversarial input model of online algorithms with revoking, the optimal deterministic algorithm for unweighted interval selection is $2k$-competitive \cite{borodin2023any}, where $k$ is the number of different interval lengths. We study this problem under random-order arrivals, a model used for beyond worst-case analysis that also captures stochastic i.i.d. settings (see Gupta \& Singla \cite{DBLP:books/cu/20/Gupta020}). While there are many instances where random arrivals help, there are problems where the competitive ratio is not significantly improved (e.g. Steiner trees where a greedy algorithm is $O(\log n)$-competitive in the worst-case, and there are $\Omega(\log n)$ bounds for both adversarial and random-order arrivals \cite{DBLP:books/cu/20/Gupta020}). We show that the simple greedy algorithm that is optimal $2k$-competitive in the adversarial case, is $2.5$-competitive in the random-order model, removing the dependence on $k$. In this model, the adversary chooses the input items, but the online sequence is a uniformly random permutation of the items. \\\\

 Finally, we use the application of interval scheduling  as motivation to begin to understand a more general issue in online algorithms, namely to  understand the  power of randomized algorithms with adversarial arrival  order compared to deterministic algorithms with random arrivals.
In this regard, we are interested to what extent can we extract random bits from the randomness in the arrival order. Specifically, we show how  to take advantage of the randomness in the arrival order to extract a random bit with bounded bias.\\\\

Interval scheduling is a well motivated and well studied problem in both the offline and online settings.  Some examples of applications related to interval scheduling are routing \cite{plotkin1995competitive}, computer wiring \cite{gupta1979optimal}, project selections during space missions \cite{hall1994maximizing}, and satellite photography \cite{gabrel1995scheduling}. We refer the reader to the surveys by Kolen et al. \cite{kolen2007interval} and Kovalyov et al. \cite{kovalyov2007fixed} for a more detailed discussion on the applications of interval scheduling.\\\\

\textbf{Related Work.} Lipton and Tomkins \cite{lipton1994online} introduced the problem of online interval scheduling. They consider the real-time setting, proportional weights, and do not allow for displacement of intervals in the solution. They give a randomized algorithm that is $O((\log\Delta)^{1+\epsilon})$-competitive, where $\Delta$ is the ratio of the longest to shortest interval. Their paper also introduced the classify and randomly select paradigm.  In the real-time unweighted setting, Faigle and Nawijn \cite{faigle1995note} consider a simple greedy $1$-competitive (optimal) deterministic algorithm with revoking. Without revoking, there is an $\Omega(n)$ lower bound both for deterministic and randomized \cite{bachmann2013online} algorithms. Woeginger \cite{woeginger1994line} considers a real-time, weighted variation of the problem with revoking, and shows that no deterministic algorithm can be constant competitive for general weights. Canetti and Irani \cite{canetti1995bounding} extend this impossibility to randomized algorithms with revoking. When an interval's weight is a function of its length, Woeginger gives an optimal $4$-competitive deterministic algorithm for special classes of weight functions. Randomized algorithms were considered for these special classes of functions \cite{seiden1998randomized,epstein2008improved}, with Fung et al. \cite{fung2014improved} currently having the best known upper bound of $2$.\\\\

In the adversarial input model, or \textit{any-order} arrivals, Bachmann et al. \cite{bachmann2013online} show a lower bound of $\Omega(n)$ for randomized algorithms in the unweighted setting without revoking. Borodin and Karavasilis \cite{borodin2023any} consider the unweighted problem with revoking, and give an optimal $2k$-competitive deterministic algorithm, where $k$ is the number of different interval lengths. This algorithm is a natural extension of the algorithm by Faigle and Nawijn \cite{faigle1995note} for any-order arrivals. Following Adler and Azar \cite{adler2003beating} who gave the first constant-competitive randomized algorithm for disjoint path allocation, Emek et al. \cite{emek2016space} give a $6$-competitive randomized algorithm for unweighted interval selection. For the case of proportional weights with revoking, Garay et al. \cite{garay1997efficient} give an optimal $(2+\sqrt5)\approx 4.23$-competitive deterministic algorithm for the problem of call control on the line, which also applies to any-order interval selection. The $2$-competitive randomized algorithm by Fung et al. \cite{fung2014improved} for the case of real-time, single-length, arbitrary weights, also applies to the any-order case.\\\\
In the random-order setting, Im and Wang \cite{im2011secretary} consider the interval scheduling secretary problem, where weighted jobs have to be processed within some interval, not necessarily continuously. They give a $O(\log D)$-competitive randomized algorithm, where $D$ is the maximum interval length of any job. More relevant to our setting, Borodin and Karavasilis \cite{borodin2023any} consider single-length unweighted interval selection with random arrivals, and show that the only deterministic memoryless algorithm\footnote{We are aware of an unpublished manuscript that proves that the one-way algorithm achieves a competitive ratio of $\approx 1.2707$ for instances that are simple chains. To the best of our knowledge, it is not known how to extend this ratio to apply to arbitrary single length instances.} that may be better than $2$-competitive, is a \textit{one-way} algorithm that replaces intervals in the same direction. Independent of our work, Garg et al. \cite{garg2024random} consider interval scheduling and maximum independent set of hyperrectangles under random arrivals. They do not allow for revoking of accepted intervals, and give a non-greedy algorithm that is \textit{strongly} (a form of high probability) $O(\log n\cdot \log\log n)$-competitive for interval selection. We note that their algorithm requires knowledge of $n$, the size of the input instance. Furthermore, they show that no algorithm that is not provided $n$ can be strongly $O(n^{1-\epsilon})$-competitive, for all $\epsilon > 0$. They also show that any greedy algorithm (i.e., one that always accepts an interval when possible) is at best $\frac{\sqrt{n}}{2}$-competitive. It follows that without knowledge of $n$, an algorithm can be forced to be greedy and hence at best $\sqrt{n}$- competitive. \\\\
Also relevant to our work is the random order algorithms (without revoking)  for the knapsack problem which is known as the knapsack secretary problem.  Babaioff et al. \cite{BabaioffIKK07} initiated the study of the knapsack secretary problem and developed a $10e$ competitive algorithm.  This was followed by the $8.06$ competitive algorithm of  Kesselheinm et al. \cite{KesselheimRTV18} and the current best $6.65$ competitive algorithm of  Albers et al. \cite{AlbersL21}. All these knapsack secretary algorithms are randomized and require knowledge of the number  of input items. The latter two papers extend their results (with somewhat worst ratios) to GAP, the generalized  assignment problem. In contrast, we will derive a much simpler deterministic online  algorithm for the knapsack problem using revoking instead of randomness. \\\\

\textbf{Our Results.}
We consider the  simple deterministic online algorithm in  \cite{borodin2023any} that is optimally $2k$-competitive in the adversarial any-order setting, and extends the $1$-competitive algorithm of \cite{faigle1995note} for the real-time setting. We analyze that algorithm under uniformly random arrivals, and we give an upper bound of $2.5$ on the competitive ratio. We use a charging argument motivated by \cite{borodin2023any} and bound the competitive ratio by the expected amount of maximum charge on any interval. We note that in contrast to the secretary algorithms,  our algorithm does not require knowledge of $n$, the length of the input sequence. We also give a lower bound of $\frac{12}{11}$ on the competitive ratio of all deterministic algorithms with revoking under random arrivals (appendix \ref{app:B}). This bound is sufficient to separate the random-order model from the real-time model, where $1$-competitiveness is attainable by a deterministic algorithm.\\\\

In a more general direction, we initiate a study of utilizing random-order arrivals to extract random bits with the goal of derandomizing algorithms. 
In particular, we consider 1-bit  barely random algorithms that randomly choose between two deterministic online algorithms $ALG_1$ and $ALG_2$  with the property that for every input sequence $\sigma$, the output of $ALG_1(\sigma) + ALG_2(\sigma) \geq OPT(\sigma) / c$ so that at least one of these two algorithms will provide a ``good''  approximation to an optimal solution on $\sigma$.  Furthermore, these deterministic algorithms usually perform optimally or near optimally when $\sigma$ is a sequence of {\it identical} items. We first present a simple process that extracts a bit with a worst case bias of $\frac{2}{3}$ when there are at least two distinct item ``types''. We then study another simple process that extracts an unbiased bit if all input items are distinct, and then a combined process that will extract a random bit with a worst case bias of $2-\sqrt{2}\approx 0.585$ under the assumption that not all items are identical. We use this last process to derandomize a number of $1$-bit barely random algorithms. In particular, we derandomize the string guessing game of \cite{BockenhauerHKKSS14}, the  Fung et al. \cite{fung2014improved} algorithms for single-length interval selection with arbitrary weights, as well as C-benevolent instances (in the terminology of Woeginger \cite{woeginger1994line}), the Han et al. \cite{han2015randomized} algorithm for the knapsack problem, the scheduling algorithm of Albers \cite{Albers02} to minimize makespan, and the unweighted throughput scheduling algorithms of Kalyanasundaram and Pruhs \cite{KalyanasundaramP03} and Chrobak et al. \cite{chrobak2007online}.  We define the {\it real time with random-order model} to mean that the set of release times in the set of input items are fixed and a random permutation is the applied to the remainder of each input item. To the best of our knowledge, with the exception of the knapsack problem, none of the aforementioned applications has been studied in the random order model. Hence our de-randomization of these 1-bit barely random algorithms provides a ``proof of concept'', constructively showing that there exists a deterministic constant competitive algorithm in the random order model.\\

\textit{Organization of the paper.} Section 2 includes definitions and a description of how the mapping from optimal intervals to intervals accepted by the algorithm is defined. We also show how the competitive ratio is bounded. Section 3 contains the main analysis of the algorithm in the random-order model. In Section 3.1 we deal with the case of two interval lengths ($k=2$). We explore the dynamics of redefining the mapping because of the displacement of intervals (revoking), and this analysis is later used to show the general case for any $k>2$ in Section 3.2. Section 4 presents three different processes to extract random bits using random-order arrivals. In Section 4.1, we apply the randomness extraction processes to derandomize a number of $1$-bit barely -random-bit algorithms. In Section 4.2 we briefly present some thought on extracting more than 1 random bit. Finally, in Section 5 we conclude with a few of the many open problems.

\section{Preliminaries}
The model consists of intervals arriving on the real line. An interval $I_{i}$ is specified by a starting point $s_{i}$, and an end point $f_{i}$, with $s_{i} < f_{i}$. It occupies space $[s_{i},f_{i})$ on the line, and the conventional notions of intersection, disjointness, and containment apply. There are two main ways intervals can conflict, and they are shown in figure \ref{fig:conflicts}. One type of conflict is  a \textit{partial conflict}, and the other type is \textit{inclusion}, or \textit{containment}. In the case of containment, we say that the smaller intervals are \textit{subsumed} by the larger one. We use $k$ to denote the number of different interval lengths of an instance. An instance with $k$ different lengths, can have a \textit{nesting depth} of at most $k-1$.\\\\
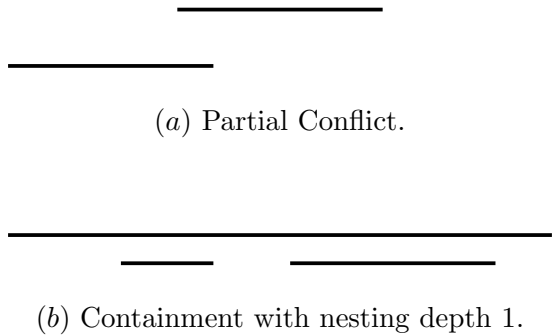
\begin{figure}[H]
	\centering
	
	\begin{tikzpicture}[scale=0.75]

	\node[draw=none] (I1a) at (-10,0) {$ $};
	\node[draw=none] (I1b) at (-6,0) {$ $};
	\draw[line width=0.5mm] (I1a) -- (I1b);


	\node[draw=none] (I2a) at (-7,1) {$ $};
	\node[draw=none] (I2b) at (-3,1) {$ $};
	\draw[line width=0.5mm] (I2a) -- (I2b);
	
	\node at (-5,-1) {$(a) \text{ Partial Conflict.}$};
	
	\node[draw=none] (I3a) at (-8,-3.5) {$ $};
	\node[draw=none] (I3b) at (-6,-3.5) {$ $};
	\draw[line width=0.5mm] (I3a) -- (I3b);

 \node[draw=none] (Ima) at (-5,-3.5) {$ $};
	\node[draw=none] (Imb) at (-1,-3.5) {$ $};
	\draw[line width=0.5mm] (Ima) -- (Imb);


	\node[draw=none] (I4a) at (-10,-3) {$ $};
	\node[draw=none] (I4b) at (-0,-3) {$ $};
	\draw[line width=0.5mm] (I4a) -- (I4b);
	
	\node at (-5,-4.5) {$(b) \text{ Containment with nesting depth 1}.$};

	\end{tikzpicture} 
	\caption{Types of conflicts.}\label{fig:conflicts}
\end{figure}
Let $OPT$ denote the size of an optimal solution, and $ALG$ the size of the algorithm's solution. We will also use $OPT$ and $ALG$ to refer respectively to an optimal solution, and the solution returned by the algorithm. The meaning should always be clear from context. We use the notion of competitive ratio to measure the performance of an online algorithm. Given an algorithm $A$ (creating a solution $ALG$), we consider the strict competitive ratio of $A$ : $CR(A) = \underset{\mathcal{I}}{\max}\frac{OPT}{\E\left[ALG\right]}$, where the expectation is over all the permutations of the input instance, and the maximum is over all input instances.\\\\ 
In our proofs, we make use of a charging argument in terms of direct and transfer charging as in \cite{borodin2023any}, although our charging argument is more significantly  involved in the random-order model. We will require a novel application of Wald's  Inequality and need to consider $OPT$ solutions that depend on the input sequence. The proof will also require a careful composition of ``base instances'' in the final analysis of arbitrary input sequences. We will now describe how the charging is done. Given an instance (set of intervals) $\mathcal{I}$ and an interval arrival sequence $\sigma$, we choose an optimal solution $OPT^{\mathcal{I}}_{\sigma}$, and define a mapping $\mathcal{F}^{\mathcal{I}}_{\sigma}: OPT^{\mathcal{I}}_{\sigma}\rightarrow ALG^{\mathcal{I}}_{\sigma}$ that shows how the intervals from an optimal solution are charged to intervals taken by the algorithm. The mapping $\mathcal{F}^{\mathcal{I}}_{\sigma}$ can be viewed as being formed and redefined throughout the execution of the algorithm as follows: On the arrival of interval $I \in OPT^{\mathcal{I}}_{\sigma}$, if $I$ is taken by the algorithm, it is mapped onto itself. If $I$ is rejected because it conflicts with some intervals taken by the algorithm, it is arbitrarily mapped to one of those conflicting intervals. Whenever an interval $I'$ is taken by replacing an existing interval $I''$, all optimal intervals mapped to $I''$ up to that point, will then be mapped to $I'$. These two first cases where optimal intervals are charged upon arrival, are instances of \textit{direct charging}. Whenever an interval is replaced by another, an instance of \textit{transfer charging} occurs to the new interval. Notice that in the end, every interval $I \in OPT^{\mathcal{I}}_{\sigma}$ is mapped to exactly one interval taken by the algorithm. We note that being able to choose a different $OPT$ for a given sequence $\sigma$, provides flexibility and facilitates our proofs. This may be important in tackling other problems in the random-order model, especially when revoking is allowed.\\\\
Given the mapping $\mathcal{F}^{\mathcal{I}}_{\sigma}$, let $\Phi : ALG\rightarrow \mathbb{Z}_{\ge 0}$ denote the charging function, which shows, at any time during the execution, the total amount of charge to any interval currently in the online algorithm's solution. That is, $\Phi (I)= |\{I' \in OPT: \mathcal{F}(I')=I\}|$. We can also express the amount of charge as $\Phi(I) = TC(I) + DC(I)$, where $TC(I)$ denotes the total amount of transfer charge to $I$ at the time it was taken by the algorithm, and $DC(I)$ denotes the total amount of direct charge to $I$.\\\\
Notice how at the end of the execution, $\sum_{I\in ALG} \Phi(I) = OPT$. We can now bound the competitive ratio of an algorithm for any instance as follows:
\begin{align*}
    \frac{OPT}{\E\left[ALG\right]}&=\frac{\E\left[\sum\limits_{1\leq i \leq ALG} \Phi(I_i)\right]}{\E\left[ALG\right]}\\\\
    &\leq \frac{\E\left[ALG\right]\max\limits_{I}\{\E\left[\Phi(I) \;|\; I \in ALG\right]\}}{\E\left[ALG\right]}\\\\
     &=\max\limits_{I}\{\E\left[\Phi(I) \;|\; I \in ALG\right]\}\\
\end{align*}
The first equality is because the sum $\Phi(I_1)+...+\Phi(I_{ALG})$ is always equal to OPT, which is a constant determined by the instance $\mathcal{I}$, and does not depend on the random arrival sequence. The inequality holds by applying Wald's inequality (as given in Young \cite{young2000k}, lemma 4.1). It follows that it suffices to bound the expected charge on every interval in $ALG$.\\
\begin{definition}[Predecessor trace]
    Let $I$ be an interval in the algorithm's final solution. The predecessor trace of $I$ is the maximal list of intervals $(P_{1},P_{2},...,P_{k}=I)$ such that $P_{i}$ was at some point accepted by the algorithm, but was later replaced by $P_{i+1}$.
\end{definition}
A predecessor trace is analogous to Woeginger's \cite{woeginger1994line} predecessor chain in the real-time model.

\section{Main Analysis for the Random-Order Model}

In this section we analyze the performance of Algorithm \ref{alg:replace-sub}. This algorithm is \textit{greedy}, in the sense that when an arriving interval does not conflict with anything, it is always accepted by the algorithm. If there are conflicts, a new interval is only accepted if it is entirely subsumed by an interval currently in the solution, which in turn gets replaced. Notice that an interval taken (maybe temporarily) by this algorithm can be directly charged by at most two optimal intervals. This is because any interval can partially conflict with at most two intervals from an optimal solution. This fact is also relevant for single-length instances ($k=1$), where no interval is replaced by this algorithm. In that case we have $TC(I) = 0$ and $DC(I)\leq 2$ for every interval $I$, giving us an upper bound of $2$ on the competitive ratio. A lower bound of $\frac{2n}{n+2}$ is given in figure \ref{fig:rom_1_lb} ($ALG = 1$ w.p. $\frac{n-2}{n}$, $OPT = 2$).
\begin{algorithm}
\caption{}\label{alg:replace-sub}
\begin{algorithmic}

\State On the arrival of $I$:
\State $I_{s} \gets $ Set of intervals currently in the solution conflicting with $I$
\If{$I_s = \emptyset$}
    \State Take $I$
    \State return
\EndIf
\For{$I' \in I_{s}$}
\If{$I \subset I'$}
    \State Take $I$ and discard $I'$
    \State return
\EndIf
\EndFor
\State Discard $I$

\end{algorithmic}
\end{algorithm}
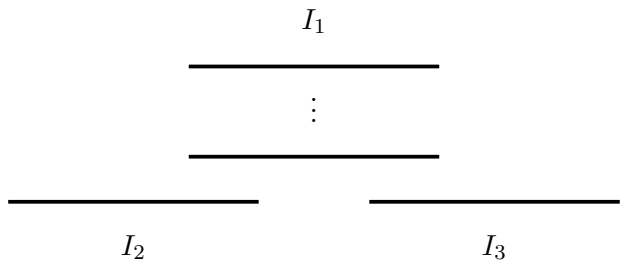
\begin{figure}[H]
	\centering
	
	\begin{tikzpicture}[scale=0.5]

	\node at (-3,1) {$I_{1}$};
	\node[draw=none] (I1a) at (-6,0) {$ $};
	\node[draw=none] (I1b) at (0,0) {$ $};
	\draw[line width=0.5mm] (I1a) -- (I1b);
	
	\node[draw=none] (d1) at (-3,-1) {$\rvdots$};
	
	\node[draw=none] (I11a) at (-6,-2) {$ $};
	\node[draw=none] (I11b) at (0,-2) {$ $};
	\draw[line width=0.5mm] (I11a) -- (I11b);

	\node[draw=none] (I2a) at (-10,-3) {$ $};
	\node[draw=none] (I2b) at (-4,-3) {$ $};
	\node at (-7,-4) {$I_{2}$};
	\draw[line width=0.5mm] (I2a) -- (I2b);

	\node[draw=none] (I3a) at (-2,-3) {$ $};
	\node[draw=none] (I3b) at (4,-3) {$ $};
	\node at (1,-4) {$I_{3}$};
	\draw[line width=0.5mm] (I3a) -- (I3b);
	
	\node[draw=none] (I2c) at (-6,-3) {$ $};
	\node[draw=none] (I3c) at (0,-3) {$ $};
	
	(I2c) -- (I2b) node [black,midway,yshift=0.5 cm]
	
	(I3a) -- (I3c) node [black,midway,yshift=0.5 cm]

	\end{tikzpicture} 
	\caption{Instance where Algorithm \ref{alg:replace-sub} is $\frac{2n}{n+2}$-competitive.}\label{fig:rom_1_lb}
\end{figure}
We will now study the case of only two interval lengths. The results of this section will later be used to show the result for $k>2$. 
\subsection{Case of $k=2$}
We first focus on a \textit{base instance} that showcases the dynamics of transfer charging. Note that in this case, any predecessor trace is of length at most two. Consider an instance with two different lengths as shown in figure \ref{fig:base-xfer}. Let $L,R,M,S$ denote the sets of corresponding intervals. The set $S$ of small intervals is entirely contained in the large intervals of $M$, and we make no assumptions about the structure of $S$. In fact, intervals in $S$ are also allowed to partially\footnote{W.l.o.g. no interval in $I_{s}\in S$ is entirely contained in an interval in $L\cup R$. If that was the case, $I_{s}$ would be considered optimal and at least one of $L,R$ would be empty.} conflict with intervals in $L\cup R$. An optimal solution consists of intervals $I_{L}\in L$, $I_{R}\in R$, and some intervals $I_{s} \subseteq S$. For the purposes of charging, we will be choosing the optimal solution that contains the latest arriving $I_{L}\in L$  and $I_{R}\in R$. The intervals in $L$ and $R$ are depicted as small intervals, but in reality they could be either small or large. We also note that intervals that are depicted as copies do not have to perfectly coincide. 

\begin{figure}[H]
	\centering
	\begin{tikzpicture}[scale=0.6]
	\filldraw[color=red!10, very thick](0,-2) ellipse (4 and 1);
    \node[color=red,thick] at (3,-3) {$S$};

	\node at (0,3) {$M$};
	\node[draw=none] (I1a) at (-8,0) {$ $};
	\node[draw=none] (I1b) at (8,0) {$ $};
	\draw[line width=0.5mm] (I1a) -- (I1b);

 \node[draw=none] (I1a1) at (-8,0.5) {$ $};
	\node[draw=none] (I1b1) at (8,0.5) {$ $};
	\draw[dotted,line width=0.5mm] (I1a1) -- (I1b1);

 \node[draw=none] (I1a2) at (-8,1) {$ $};
	\node[draw=none] (I1b2) at (8,1) {$ $};
	\draw[dotted,line width=0.5mm] (I1a2) -- (I1b2);

 \node[draw=none] (d1) at (0,2) {$\vdots$};
	
	\node[draw=none] (I2a) at (-2.5,-2) {$ $};
	\node[draw=none] (I2b) at (-0.5,-2) {$ $};
	\draw[line width=0.5mm] (I2a) -- (I2b);
	
	\node[draw=none] (I3a) at (-1,-1.5) {$ $};
	\node[draw=none] (I3b) at (1,-1.5) {$ $};
	\draw[line width=0.5mm] (I3a) -- (I3b);
	
	\node[draw=none] (I4a) at (0.5,-2) {$ $};
	\node[draw=none] (I4b) at (2.5,-2) {$ $};
	\draw[line width=0.5mm] (I4a) -- (I4b);
	
	\node[draw=none] (d1) at (-3,-2) {$\dots$};
    \node[draw=none] (d1) at (3,-2) {$\dots$};

	\node[draw=none] (ILa) at (-9.5,-0.5) {$ $};
	\node[draw=none] (ILb) at (-7.5,-0.5) {$ $};
	\node at (-8.5,0) {$L$};
	\draw[line width=0.5mm] (ILa) -- (ILb);
 \node[draw=none] (ILa1) at (-9.5,-1) {$ $};
	\node[draw=none] (ILb1) at (-7.5,-1) {$ $};
	\draw[dotted,line width=0.5mm] (ILa1) -- (ILb1);
 \node[draw=none] (d1) at (-8.5,-1.6) {$\vdots$};

 \node[draw=none] (IRa) at (7.5,-0.5) {$ $};
	\node[draw=none] (IRb) at (9.5,-0.5) {$ $};
	\node at (8.5,0) {$R$};
	\draw[line width=0.5mm] (IRa) -- (IRb);
    \node[draw=none] (IRa1) at (7.5,-1) {$ $};
	\node[draw=none] (IRb1) at (9.5,-1) {$ $};
	\draw[dotted,line width=0.5mm] (IRa1) -- (IRb1);
\node[draw=none] (d1) at (8.5,-1.6) {$\vdots$};
	\end{tikzpicture} 
	\caption{Base instance for transfer charging with $k=2$.
}\label{fig:base-xfer}
\end{figure}
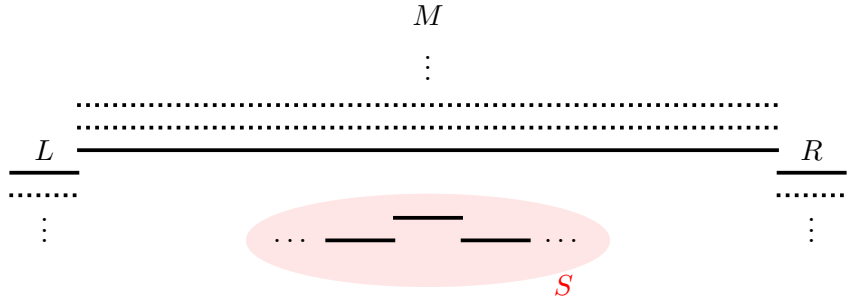

\begin{lemma}\label{lemm:k_2_main}
    For any instance with a structure as depicted in figure \ref{fig:base-xfer}, we have that\\ $\E[\Phi(I) \; | \; I \in ALG] \leq 2.5$.
\end{lemma}
\begin{proof}
    We will be writing $\E[\Phi(I)]$ for readability. We have that $\Phi(I) = TC(I) + DC(I)$. As mentioned before, $\forall I, DC(I) \leq 2$. We will now bound $\E[TC(I)]$. Let $TC_{1}$ denote the event that a transfer charge of $1$ occurs, and $TC_{2}$ denote the event that a transfer charge of $2$ occurs. We focus on the first arrival of an interval from $S$, as that interval will receive the transfer charge. Let $N = |L| + |R| + |M| + |S|$.\\\\

    We want to compute: $\underset{|L|,|R|,|M|,|S|}{argmax(\E[TC(I)])}$, where $\E[TC(I)] = Pr(TC_{1}) + 2Pr(TC_{2})$.\\\\

    \textit{Case of $TC_{2}$:} For a transfer charge of $2$ to occur, it must be that an interval from $M$ arrives first, and that all the intervals in $L \cup R$ arrive before the first interval from $S$. This is an experiment of drawing without replacement, and the probability that we get all intervals in $L\cup R$ before the first interval from $S$ is the following:
    $$\frac{|L|+|R|}{|L|+|R|+|S|} \cdot \frac{|L|+|R|-1}{|L|+|R|+|S|-1} \dots \frac{1}{|S| + 1} = \frac{(|L|+|R|)!\cdot |S|!}{(|L|+|R|+|S|)!}$$\\
    and therefore, $$Pr(TC_{2}) = \frac{|M|}{N}\cdot \frac{(|L|+|R|)!\cdot |S|!}{(|L|+|R|+|S|)!}$$\\\\

    \textit{Case of $TC_{1}$:} For a transfer charge of $1$ to occur, it must be that an interval from $M$ arrives first, and then one of two cases: all intervals from $L$ (respectively $R$) arrive, followed by the first interval of $S$, and the last interval of $R$ (respectively $L$) arrives after. These two cases are symmetrical and we'll focus on the first one, which can be visualized as follows:
    $$\text{first } M \rightarrow \text{last } L \rightarrow \text{first } S \rightarrow \text{last } R$$\\
    Consider the following two events:\\
    \textit{Event $A_L$:} The first interval from $S$ arrives after the last interval from $L$.\\
    \textit{Event $B_R$:} The last interval from $R$ arrives after the first interval from $S$.\\

    We want to compute $Pr(A_L \cap B_R)$. Notice that in the previous case of $TC_2$ we computed $Pr(A_L \cap \overline{B_R})$. We get that:
    $$Pr(\overline{B_R} |A_L) = \frac{Pr(A_L \cap \overline{B_R})}{Pr(A_L)} = \frac{(|L| + |R|)! \cdot (|L| + |S|)!}{(|L| + |R| + |S|)!\cdot |L|!}$$\\\\
    \begin{align*}
     Pr(A_L \cap B_R)&= Pr(B_R|A_L)\cdot Pr(A_L) = (1-Pr(\overline{B_R}|A_L))\cdot Pr(A_L)\\\\
    &= \left(1-\frac{(|L|+|R|)!\cdot(|L|+|S|)!}{(|L|+|R|+|S|)!\cdot |L|!}\right)\cdot \frac{|L|!\cdot |S|!}{(|L|+|S|)!}\\\\
    &= \frac{|L|!\cdot |S|!}{(|L|+|S|)!} - \frac{(|L|+|R|)!\cdot |S|!}{(|L|+|R|+|S|)!}
\end{align*} 
Similarly, for the symmetrical case we get that:
$$Pr(A_R\cap B_L) = \frac{|R|!\cdot |S|!}{(|R|+|S|)!} - \frac{(|L|+|R|)!\cdot |S|!}{(|L|+|R|+|S|)!}$$
Combining the two cases we get that:
\begin{align*}
     Pr(TC_1)&=\frac{M}{N} \left[Pr(A_L \cap B_R) + Pr(A_R \cap B_L)\right]\\\\
    &=\frac{M}{N} \left[\frac{|L|!\cdot |S|!}{(|L|+|S|)!} + \frac{|R|!\cdot |S|!}{(|R|+|S|)!} -2\frac{(|L|+|R|)!\cdot |S|!}{(|L|+|R|+|S|)!}\right] \\\\
\end{align*}
    Finally, we have that:
    $$\E[TC(I)] = Pr(TC_{1}) + 2Pr(TC_{2}) =  \frac{M}{N} \left[\frac{|L|!\cdot |S|!}{(|L|+|S|)!} + \frac{|R|!\cdot |S|!}{(|R|+|S|)!} \right]$$
    We have that $(|L|,|R|,|S|,|M|) \in \{\mathbb{N}^{*}\}^{4}$. We will also assume that $|S|\geq 3$. We deal with the cases of $|S| = 1$ and $|S|=2$ in appendix \ref{app:A}. To maximize $\E[TC(I)]$ we set $|L| = |R| = 1$, and $|S| = 3$, and we get:
  $$\E[TC(I)] \leq \frac{|M|}{|M|+5}\cdot \frac{2 \cdot 3!}{4!} = \frac{1}{2}\frac{|M|}{|M|+5} \xrightarrow{|M|\rightarrow +\infty}\frac{1}{2}$$\\\\
  Therefore, $\E[\Phi(I)] = \E[TC(I) + DC(I)] \leq \frac{1}{2} + 2$.
\end{proof}
\vspace{5mm}
\begin{corollary}
    Algorithm \ref{alg:replace-sub} is $2.5-$competitive on instances of the form as in figure \ref{fig:base-xfer}.\\
\end{corollary}

Of course, so far we have only analyzed what may seem like very restrictive base instances. 
It turns out the base instances are quite representative in the sense that we can show how any input sequence can be decomposed into base instances.
We first prove the following Theorem for the general case of $k=2$. \\
\begin{theorem}
    Algorithm \ref{alg:replace-sub} achieves a competitive ratio of $2.5$ for the problem of interval selection on instances with at most two different lengths.
\end{theorem}

\begin{proof}
Let $C_{i} = (L_{i},R_{i},M_{i},S_{i})$ denote a basic construction (or sub-instance) that follows the structure described earlier, with $L_i\cup R_i \neq \emptyset$. Given an optimal solution $OPT$, any instance can be partitioned into a set of such constructions, each being uniquely identified by its middle non-optimal intervals. Let $\mathcal{C} = \{C_{1},C_{2},...,C_{n}\}$ denote the set of all these constructions. Although it could be that $|M_{i}| \gg 1$, we will abuse the notation and refer to the \textit{interval $M_{i}$}. Figure \ref{fig:basic_cons_conflicts} shows an instance that is partitioned into three basic constructions: $C_1 = (\emptyset,\{I_1\},\{M_1\},\{J_1,J_2,J_3\})$, $C_2 = (\{I_1\},\{I_3\},\{M_2\},\{I_2,J_4\})$, and $C_3=(\{I_2\},\{I_4\},\{M_3\},\{J_4,I_3\})$.\\\\
\begin{figure}[H]
	\centering
	
	\begin{tikzpicture}[scale=0.5]
	
	\node at (-8.5,-4) {$J_1$};
	\node[draw=none] (J1a) at (-9.7,-3.5) {$ $};
	\node[draw=none] (J1b) at (-7.3,-3.5) {$ $};
	\draw[line width=0.5mm] (J1a) -- (J1b);
 
	\node at (-7,-4.5) {$J_2$};
	\node[draw=none] (J2a) at (-8.2,-4) {$ $};
	\node[draw=none] (J2b) at (-5.8,-4) {$ $};
	\draw[line width=0.5mm] (J2a) -- (J2b);
 
	\node at (-5.5,-4) {$J_3$};
	\node[draw=none] (J3a) at (-6.7,-3.5) {$ $};
	\node[draw=none] (J3b) at (-4.3,-3.5) {$ $};
	\draw[line width=0.5mm] (J3a) -- (J3b);

    \node at (0.5,-3) {$I_1$};
    \node[draw=none] (I1a) at (-0.7,-3.5) {$ $};
	\node[draw=none] (I1b) at (1.7,-3.5) {$ $};
	\draw[line width=0.5mm] (I1a) -- (I1b);

	\node[draw=none] (M1a) at (-10,-3) {$ $};
	\node[draw=none] (M1b) at (-0,-3) {$ $};
	\node at (-5,-2.5) {$M_1$};
	\draw[line width=0.5mm] (M1a) -- (M1b);

    \node[draw=none] (M2a) at (1,-3) {$ $};
	\node[draw=none] (M2b) at (11,-3) {$ $};
	\node at (6,-2.5) {$M_2$};
	\draw[line width=0.5mm] (M2a) -- (M2b);

    \node at (11.5,-3) {$I_3$};
    \node[draw=none] (I3a) at (10.3,-3.5) {$ $};
	\node[draw=none] (I3b) at (12.7,-3.5) {$ $};
	\draw[line width=0.5mm] (I3a) -- (I3b);

    \node[draw=none] (M3a) at (5,-4.5) {$ $};
	\node[draw=none] (M3b) at (15,-4.5) {$ $};
	\node at (10,-4) {$M_3$};
	\draw[line width=0.5mm] (M3a) -- (M3b);

    \node at (4.5,-4.5) {$I_2$};
    \node[draw=none] (I2a) at (3.3,-4) {$ $};
	\node[draw=none] (I2b) at (5.7,-4) {$ $};
	\draw[line width=0.5mm] (I2a) -- (I2b);

    \node at (15.5,-3.5) {$I_4$};
    \node[draw=none] (I4a) at (14.3,-4) {$ $};
	\node[draw=none] (I4b) at (16.7,-4) {$ $};
	\draw[line width=0.5mm] (I4a) -- (I4b);

    \node at (7,-5.5) {$J_4$};
    \node[draw=none] (J4a) at (5.8,-5) {$ $};
	\node[draw=none] (J4b) at (8.2,-5) {$ $};
	\draw[line width=0.5mm] (J4a) -- (J4b);

	\end{tikzpicture} 
	\caption{Example instance with three basic constructions.}\label{fig:basic_cons_conflicts}
\end{figure}
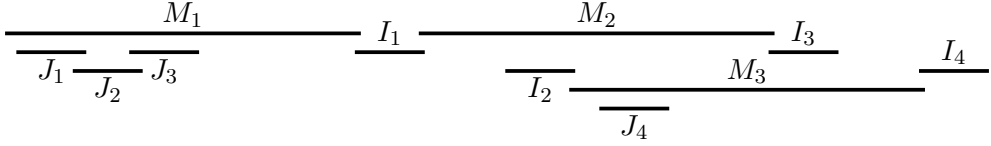

    We focus on these constructions, because a transfer charge can only occur to the 
    intervals in $(S_{1} \cup S_{2} \cup ... \cup S_{n})$. It is helpful to associate the event of a transfer charge with the related construction, and note that no transfer charge will be associated with that construction again. For example, after interval $J_4$ is taken (fig. \ref{fig:basic_cons_conflicts}), no transfer charge can be associated with $C_2$ and $C_3$. This is because whenever an interval from $S_{i}$ is accepted, it can never be replaced again, and $M_{i}$ cannot be accepted again. 
    Consider the middle intervals $\{M_{1},...,M_{n}\}$. A subset of those intervals will be taken by the algorithm during the execution. Whenever $M_{i}$ is taken, a transfer charge may occur. Intervals outside $C_{i}$ can affect it by blocking intervals in $L_{i} \cup R_{i}$ (before or after $M_{i}$ is taken), or $S_{i}$ (before $M_{i}$ is taken). In every case, the expected total amount of charge on any $S_{i}$ is no more compared to having $C_{i}$ separately, which, in addition to the fact that any interval is transfer charged at most once, gives us the desired result.\\\\
\end{proof}

\subsection{Case of $k>2$}
In the case of $k>2$, the nesting depth can be greater than $1$, which allows for a predecessor trace of length greater than $2$. As before, we fix an optimal solution and consider the set of all basic constructions $\mathcal{C}$ derived from the instance. We note that a basic construction can now be contained in another. More specifically, if $C_{i}$ is contained in $C_{j}$, it means that $(L_{i} \cup R_{i} \cup M_{i} \cup S_{i})\subseteq S_{j}$. For every interval $J \in ALG$, we consider the predecessor trace $P_{J} = (M_{1},...,M_{d},J)$ that was formed during the execution. When an interval ($M_{i}$) is replaced by another ($M_{i+1}$), it also transfers all of its charge, and we have that $\Phi(M_{i})\leq \Phi(M_{i+1})\leq \Phi(M_{i}) + 2$. W.l.o.g. we assume that $M_{1},...,M_{d}$ are all middle intervals of some basic constructions. The only way this isn't true is if for some $i$, $M_{i}$ does not partially conflict with any optimal intervals, in which case it cannot increase the amount of charge transferred to $M_{i+1}$. Interval $J$ may or may not correspond to a middle interval, but we know that $J\in S_{d}$. The total amount of charge on any one interval is a random quantity, and we want to upper bound $\E[\Phi(J)\; | \; J \in ALG]$ for every predecessor trace $P_{J}$.\\
\begin{lemma}\label{lemm:big_k_main}
    $\E[\Phi(J) \; | \; J \in ALG] \leq 2.5$ for every $J\in ALG$ and predecessor trace $P_J$.\\
\end{lemma}
\begin{proof}
As before, we have that $DC(J)\leq 2$ and we focus on $\E[TC(J)]$. Notice that for every $P_J$, the expected amount of charge added to $J$ because of interval $M_i$ depends on the subset of intervals of $(L_i \cup R_i \cup S_i)$ that are yet to arrive after $M_i$ was accepted. We are able to derive a bound on $\E[TC(J)]$ by assuming the last interval of $L_i$ and the last interval of $R_i$ are yet to arrive, and lower bounding the remaining of $S_i$ by the fact that intervals $M_{i+1},...,J$ are yet to arrive.\\\\
Let $(S'_i, L'_i, R'_i) \subseteq (S_i, L_i, R_i)$ be the set of intervals pending to arrive after $M_{i}$ was accepted, and for $i>1$, let $M'_i \subseteq M_i$ be the set of intervals pending to arrive after $M_{i-1}$ was accepted. Similarly, let $M'_1 \subseteq M_1$ be the set of intervals pending after $M_1$ was able to be accepted by the algorithm. For readability, we omit $|\cdot|$ and refer to $|M'_{i}|, |S'_{i}|, |L'_{i}|$, and $|R'_{i}|$ as $M'_{i}, S'_{i}, L'_{i}, R'_{i}$. Let $N'_{i} = M'_{i} + L'_{i} + R'_{i} +S'_{i}$. For $1\leq i\leq d$ we have that:
\[
S'_{i} \geq S'_{d} + \sum_{j=i+1}^{d} (M'_{j}+L'_{j}+R'_{j}) \tag{1}
\]
We first consider the case of $S'_{d} \geq 3$. From the analysis of Lemma \ref{lemm:k_2_main} we get that:
\begin{align*}
\E[TC(J)] &= \frac{M'_{1}}{N'_{1}}\cdot \left[\frac{L'_{1}! \cdot S'_{1}!}{(L'_{1}+S'_{1})!} + \frac{R'_{1}! \cdot S'_{1}!}{(R'_{1}+S'_{1})!}\right] +\dots + \frac{M'_{d}}{N'_{d}}\cdot \left[\frac{'L_{d}! \cdot S'_{d}!}{(L'_{d}+S'_{d})!} + \frac{R'_{d}! \cdot S'_{d}!}{(R'_{d}+S'_{d})!}\right]\\\\
    & \leq \frac{M'_{1}}{M'_{1}+2+S'_{1}} \cdot \left[ 2\frac{S'_{1}!}{(S'_{1} + 1)!} \right] + \dots + \frac{M'_{d}}{M'_{d}+2+S'_{d}} \left[ 2\frac{S'_{d}!}{(S'_{d}+1)!} \right]\\\\
    & = 2\left[\frac{M'_1}{(M'_1+2+S'_1)\cdot(S'_1+1)} + \dots + \frac{M'_d}{(M'_d + 2 + S'_d)\cdot (S'_d+1)}\right]\\\\
    &\leq 2\left[ \frac{M'_{1}}{5+2(d-1)+\sum_{j=1}^{d}M'_{j}}\cdot \frac{1}{4+2(d-1)+\sum_{j=2}^{d}M'_{j}} + \dots + \frac{M'_{d}}{M'_{d}+5}\cdot\frac{1}{4} \right]\\\\
    &= 2\sum_{i=1}^{d} \frac{M'_{i}}{(5+2(d-i)+\sum_{j=i}^{d}M'_{j})\cdot (4+2(d-i)+\sum_{j=i+1}^{d}M'_{j})}\\\\
\end{align*}
The first inequality is because we set $L'_{i}=R'_{i}=1$. The second inequality is because of $(1)$.\\\\
Let $F^{d}(x_1,..,x_d) =\sum\limits_{i=1}^{d} \frac{x_{i}}{(5+2(d-i)+\sum_{j=i}^{d}x_{j})\cdot (4+2(d-i)+\sum_{j=i+1}^{d}x_{j})}$, with $x_i \geq 1$ for every $i$. For readability, we will write $F^d(x_1,...,x_d) = \sum\limits_{i=1}^{d} \frac{x_{i}}{(x_i + 2 + s_i)\cdot (s_i + 1)}$, with $s_i = x_{i+1} + 2 + s_{i+1}$, and $s_d = 3$. We will show that for any $d$ and $x_1,...,x_d$:
$$F^{d}(x_1,...,x_d) \leq \frac{1}{4}$$
We show this by induction on $d$:\\\\
\textit{Base case $d=1$}: $F^{1}(x_1)=\frac{x_1}{4(x_1+5)} \leq \frac{1}{4}$ holds.\\\\
\textit{Induction step}: For $d=D$, we assume that $F^D(x_1,..,x_D) \leq \frac{1}{4}$.\\\\
For $d=D+1$, we focus on the first two terms of the sum:
\begin{align*}
    \sum\limits_{i=1}^{2} \frac{x_{i}}{(x_i + 2 + s_i)\cdot (s_i + 1)}&=\frac{x_1}{(x_1 + 2 + s_1)\cdot (s_1 + 1)} + \frac{x_2}{(x_2 + 2 +s_2)\cdot (s_2 + 1)}\\\\
    &=\frac{x_1}{(x_1 + x_2 + 4 + s_2)\cdot (x_2 + 3 + s_2)} + \frac{x_2}{(x_2 + 2 +s_2)\cdot (s_2 + 1)}
\end{align*}
We will show that:
\[
\sum\limits_{i=1}^{2} \frac{x_{i}}{(x_i + 2 + s_i)\cdot (s_i + 1)}\leq \frac{x_1 + x_2}{(x_1 + x_2 + 2 + s_2)\cdot (s_2 + 1)} \tag{2}
\]
We have that:
\begin{align*}
    \frac{x_1 + x_2}{(x_1 + x_2 + 2 + s_2)\cdot (s_2 + 1)} - \frac{x_1}{(x_1 + x_2 + 4 +    s_2)\cdot (x_2 + 3 + s_2)} - \frac{x_2}{(x_2 + 2 +s_2)\cdot (s_2 + 1)} &\geq\\\\    
    \frac{x_1}{(x_1 + x_2 + 2 + s_2)\cdot (s_2 + 1)} - \frac{x_1}{(x_1 + x_2 + 4 +    s_2)\cdot (x_2 + 3 + s_2)}&\geq 0
\end{align*}
Therefore $(2)$ holds, and we have that $F^{D+1}(x_1,x_2,..,x_{D+1})\leq F^{D}(x_1+x_2,x_3,..,x_{D+1})$, which we know is at most $\frac{1}{4}$ by the induction hypothesis.\\\\
Putting everything together, we have that $\E[TC(J)]\leq \frac{1}{2}$, and because $DC(J) \leq 2$, we get that $\E[\Phi(J)]\leq 2.5$. The cases of $s_d = 1$ and $s_d = 2$ are dealt with in appendix \ref{app:A}.
\end{proof}

\begin{corollary}
    Algorithm \ref{alg:replace-sub} is $2.5$-competitive for the problem of interval selection for all $k$.
\end{corollary}

\section{Randomness Extraction}
\label{sec:randomness-extraction}

Our analysis of the competitiveness for unweighted interval selection is another example of the power of random-order arrivals. This leads us to a  basic question with regard to  online algorithms; namely, what  is the power of random-order deterministic algorithms relative to adversarial order randomized algorithms. We know that there are a number of problems where deterministic random-order algorithms provide provably better (or at least as good) competitive ratios than randomized algorithms with adversarial order (e.g., prominently, the secretary problem and bipartite matching \cite{Mehta2013}). But are there problems where adversarial order randomized algorithms are provably (or at least seemingly) better than random-order deterministic algorithms? It is natural then to see if we can use the randomness in the arrival of input items to extract random bits. Such bits may then  be used to derandomize certain algorithms when assuming random-order arrivals. \textit{Barely random} algorithms \cite{reingold1994randomized} use a (small) constant number of random bits, and are well suited to be considered for this purpose. These algorithms are often used in the \textit{classify and randomly select} paradigm, where inputs are partitioned into a small number of classes, and the algorithm randomly selects a class of items to work with. 

Our focus will be on simple $1$-bit randomness extraction processes which we can then apply to $1$-bit barely random algorithms. We briefly consider a way of extracting two biased bits for derandomizing two-length interval selection with arbitrary weights, and discuss possible extensions to extracting multiple bits from random order sequences.\\\\ 



 Our first randomness extraction process is described in Process \ref{alg:proc-1}. We assume there exist at least two different classes, or item types, that all input items belong to. Furthermore, we maintain a counter that represents the number of items that have arrived so far. Our process returns $1$, if the first item of the second type to arrive online is when the counter is even. The adversary can choose the number of items of each type to affect the probability of getting either output, but we show that the worst case bias of this bit is at most $\frac{2}{3}$.

\begin{process}
\caption{Biased bit extraction}\label{alg:proc-1}
\begin{algorithmic}
\State On the arrival of $I_i$:
\If{$i = 1$}
    \State $type \gets type(I_1)$
    \Else
        \If{$type \neq type(I_i)$}
            \State return($1-(i \mod 2)$) and terminate the process
        \EndIf
\EndIf

\end{algorithmic}
\end{process}
\begin{theorem}
    The bit extracted by process \ref{alg:proc-1} has a worst case bias of at most $\frac{2}{3}$. That is, $Pr(bit=1) \in (\frac{1}{2},\frac{2}{3})$.
\end{theorem}
\begin{proof}
    Let there be $A$ items of $type_{A}$ and $B$ items of $type_{B}$, with $N = A + B$. Let $E_{v}$ be the event where the second item type arrives on an even counter. Let also $A_{e}$ (respectively $B_{e}$) be the event that the first appearance of $type_{A}$ ($type_{B}$) is on an even counter, and $F_{A}$ (resp. $F_{B}$) be the event that the very first item that arrives is of $type_{A}$.
    We assume that $N$ is very large, and we are sampling from an infinite population. We have that:
    $$Pr(E_{v})= Pr(B_{e}|F_{A})\cdot Pr(F_{A}) +  Pr(A_{e}|F_{B})\cdot Pr(F_{B})$$
    We start by computing $Pr(B_{e}|F_{A})$:
    $$ Pr(B_{e}|F_{A})= \frac{B}{N} + \left(\frac{A}{N}\right)^{2} \cdot \frac{B}{N} + \left(\frac{A}{N}\right)^{4} \cdot \frac{B}{N} + \left(\frac{A}{N}\right)^{6} \cdot \frac{B}{N} +\dots = \frac{B}{N}\sum_{i=0}^{+\infty}\left(\frac{A}{N}\right)^{2i}$$
and therefore:
$$Pr(B_{e}|F_{A})\cdot Pr(F_{A}) = \frac{AB}{N^{2}}\sum_{i=0}^{+\infty}\left(\frac{A}{N}\right)^{2i}$$
Similarly, we get that: 
$$Pr(A_{e}|F_{B})\cdot Pr(F_{B}) = \frac{AB}{N^{2}}\sum_{i=0}^{+\infty}\left(\frac{B}{N}\right)^{2i}$$
Putting everything together:
$$Pr(E_{v})= \frac{AB}{N^{2}}\sum_{i=0}^{+\infty}\left(\frac{A^{2i} + B^{2i}}{N^{2i}}\right)$$
Let $A=\alpha N$ with $\alpha \in (0,1)$. We can rewrite $Pr(E_{v})$ as follows:
\begin{align*}
    Pr(E_{v})&=\frac{\alpha(1-\alpha) N^{2}}{N^{2}} \sum_{i=0}^{+\infty}\frac{\left(\alpha N\right)^{2i}+(1-\alpha)^{2i}N^{2i}}{N^{2}}\\\\
    &= \alpha (1-\alpha) \sum_{i=0}^{+\infty} \alpha^{2i} + (1-\alpha)^{2i}\\\\
    &= \alpha (1-\alpha) \frac{-2\alpha^2 + 2\alpha +1}{(\alpha - 2)\alpha(\alpha^2-1)}\\\\
    &= \frac{2\alpha^2 -2\alpha -1}{(\alpha +1)(\alpha -2)}
\end{align*}
Let $f(\alpha) = \frac{2\alpha^2 -2\alpha -1}{(\alpha +1)(\alpha -2)}$. We have that $f[(0,1)] \in (\frac{1}{2},\frac{2}{3})$ (figure \ref{fig:biased}). In conclusion, the worst case bias of the bit extracted through Process \ref{alg:proc-1}  is $\frac{2}{3}$.
\begin{figure}[h]
        \includegraphics[scale=0.6]{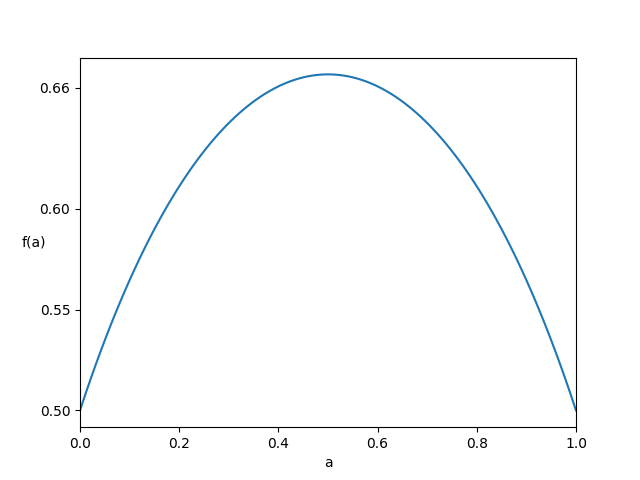}
        \centering
        \caption{Plot of $f(a)$}\label{fig:biased}
      \end{figure}
\end{proof}

Fung et al. \cite{fung2014improved} give barely random algorithms for some weighted variations of interval selection in the real-time model. Their $2$-competitive algorithm for the case of single-length intervals with arbitrary weights is also directly applicable to the any-order model, and consequently to the random-order model, while maintaining the same competitiveness. The real line is divided into slots, and each interval can be viewed as belonging to an \textit{even} or \textit{odd} slot. Their algorithm uses one random bit to pick one slot type, and gets an optimal solution amongst those intervals. We refer to \cite{fung2014improved} for a complete description of the algorithm. We can derandomize this algorithm using a random bit extracted by Process \ref{alg:proc-1} as follows: Our algorithm starts working on the first type of intervals that arrive as if it was the chosen one. When a new (slot-)type arrives, our bit is extracted, and we decide whether we will switch to the second type intervals or not. With a slight change in their proof (Theorem 3.1 in \cite{fung2014improved}) we see that our algorithm is $3$-competitive.\\\\

We note that revoking is essential in the above algorithm. Although the $2$-competitive algorithm by Fung et al. already requires that revoking is allowed in the model, we also need to be able to discard the entire solution constructed by the time the random bit is extracted. Process \ref{alg:proc-1} may be used more generally to derandomize algorithms that fall under the classify and randomly select paradigm, when two classes are used. Consider a deterministic algorithm $A$, and let $ALG_1$ (resp. $ALG_2$) denote the performance of the algorithm on input items that belong to class 1 (resp. 2). Assuming $ALG_1 +ALG_2 \geq \frac{OPT}{c}$, with $c\geq1$, we get a $3c$-competitive algorithm.\\\\

In all of our processes and applications we naturally assume that input items are represented by vectors ${\bf x} = (x_1, \ldots , x_d)  \in {\mathbb R}^d$ for some $d \geq 1$. This allows us to strictly order distinct (i.e. not identical) items. That is, ${\bf x} < {\bf y}$ iff $x_i = y_i$ for $i < k$ and $x_k < y_k$ for some $k$ with $1 \leq k \leq d$. 
Under the assumption that all input items are distinct (and hence there exists a global ordering),  we are able to extract an unbiased bit simply by comparing the two first items to arrive (Process \ref{alg:proc-2}). Note that while this is generally a strong assumption, we can envision applications where this may be reasonable to assume (e.g. interval selection with distinct intervals). Notice that we can repeat this Process \ref{alg:proc-2} for the next pair of online items. Generally, given $2N$ online items we can extract $N$ unbiased bits.


\begin{process}
\caption{Distinct-Unbiased}\label{alg:proc-2}
\begin{algorithmic}
\State On the arrival of $I_1,I_2$:

\If{$I_1 < I_2$}
    \State return(1)
\Else
    \State return(0)
\EndIf

\end{algorithmic}
\end{process}
\begin{theorem}
    Under the assumption that all input items are distinct (and hence there exists a global ordering amongst all items in the input instance), process \ref{alg:proc-2} (Distinct-Unbiased) produces an unbiased bit.
\end{theorem}
\begin{proof}
    Let $I_1, I_2, ..., I_N$ be all the items in the instance, such that $I_1 < I_2 <...<I_N$. Let $I_a, I_b$ be the first two items that arrive. Let $E_1$ denote the event that $I_a < I_b$. Let $F_i$ denote the event that item $I_i$ arrives first, and $S_i$ denote the event that $I_{b} > I_{i}$. We have that:
    \begin{align*}
        Pr(E_1)&=Pr(S_1 | F_1)\cdot Pr(F_1)+ \dots + Pr(S_N | F_N)\cdot Pr(F_N)\\\\
        &=\frac{1}{N}\frac{N-1}{N-1} + \dots + \frac{1}{N}\frac{N-N}{N-1}\\\\
        & = \frac{1}{N}\sum_{i=1}^{N}\frac{N-i}{N-1} = \frac{1}{2}
    \end{align*}
\end{proof}


Motivated by the previous two extraction processes, we arrive at an improved $1$-bit extraction process. Under the mild assumption that there exist at least two distinct items, and we have an ordering amongst different items, we propose the following process: If the first two input items are identical, we use Process \ref{alg:proc-1}, otherwise we use the ordering of the types of the first two items to decide on the bit (return $1$ if the second type is less than the first, else $0$). We call this process \texttt{COMBINE}. The item types used by Process \ref{alg:proc-1} are defined by the first item that arrives. Items that are identical to the first item are of the same type, while any distinct item is of a different type. We will show the worst-case bias of the returned bit is $2-\sqrt{2} \approx 0.585$. Let $\mathcal{I}$ be the set of input items, and $\sigma: \mathcal{I} \xrightarrow[]{} \mathbb{N}$ be the function dictating the ordering of different items. Furthermore, let $\mathbb{D}$ be the event the first two input items are distinct, and $\mathbb{I}$ the event the first two items are identical. We start by proving the following lemma:
\begin{lemma} \label{lem:half-distinct}
    $Pr(\sigma(I_2) < \sigma(I_1) \;|\; \mathbb{D}) = \frac{1}{2}$, where $I_1,I_2$ are the first two items to arrive online.
\end{lemma}
\begin{proof}
    Let $d \leq N$ be the number of distinct items (or types), and let $C_1, C_2,...,C_d$ be the number of copies of each type, with $\sigma(type_i) < \sigma(type_j)$ when $i<j$. As before, we assume $N$ is very large and we are sampling from an infinite population. We have that:
    $$Pr(\sigma(I_2) < \sigma(I_1)) = \sum_{i=1}^{d} \sum_{j = 1}^{i-1}\frac{C_i}{N}\frac{C_j}{N}$$
    Furthermore:
    $$Pr(\mathbb{D}) = \sum_{i = 1}^{d} \sum_{\substack{j = 1\\ j\neq i}}^{d}\frac{C_i}{N}\frac{C_j}{N} = 2 \sum_{i = 1}^{d} \sum_{\substack{j = 1}}^{i-1}\frac{C_i}{N}\frac{C_j}{N} = 2 \cdot Pr(\sigma(I_2) < \sigma(I_1))$$
    This is because every pair $(i,j), i\neq j$ appears twice in the summation of $Pr(\mathbb{D})$. In conclusion:
    $$Pr(\sigma(I_2) < \sigma(I_1) \;|\; \mathbb{D}) = \frac{Pr(\sigma(I_2) < \sigma(I_1))}{Pr(\mathbb{D})} = \frac{1}{2}$$
\end{proof}
We will now prove the following theorem.
\begin{theorem}
    The bit returned by Process \texttt{COMBINE} has a worst-case bias of $2-\sqrt{2}$. 
\end{theorem}
\begin{proof}
    Let the first item to arrive online occur with frequency $r\in (0,1)$, i.e. there are $rN$ copies of that type in the input. Let $\mathbb{E}_2$ be the event that Process \ref{alg:proc-1} returns $1$. We get that:
    \begin{align*}
        Pr(bit = 1) &= Pr(bit = 1 \cap \mathbb{D}) + Pr(bit = 1 \cap \mathbb{I})\\
        &= Pr(\sigma(I_2) < \sigma(I_1) \;|\; \mathbb{D}) Pr(\mathbb{D}) + Pr(\mathbb{E}_2 \; | \; \mathbb{I}) Pr(\mathbb{I})\\
        &= \frac{1}{2}(1-r) + r \cdot Pr(\mathbb{E}_2 \; | \; \mathbb{I}) 
    \end{align*}
    where the last equality follows from Lemma \ref{lem:half-distinct}.\\
    Furthermore:
    \begin{align*}
      Pr(\mathbb{E}_2 \; | \; \mathbb{I}) &= (1-r) + r^2(1-r) + r^4(1-r) +...\\
      &= (1-r)\sum_{i=0}^{\infty} r^{2i}\\
      &= (1-r)\frac{1}{1-r^2} = \frac{1}{1+r}
    \end{align*}
    Putting everything together:
    $$Pr(bit = 1) = \frac{1}{2}(1-r) + \frac{r}{1+r}$$
    which achieves a maximum value of $2-\sqrt{2}\approx 0.585$ at $r = \sqrt{2}-1 \approx 0.414$. More specifically, we have that $Pr(bit = 1) \in (\frac{1}{2},2-\sqrt{2})$, which concludes the proof.
\end{proof}

\subsection{Applications of the 1-bit extraction method \texttt{COMBINE}}
\textit{Single-length interval selection}. The algorithm of Fung et al. \cite{fung2014improved} for the problem of single-length interval selection with arbitrary weights that we derandomized using Process \ref{alg:proc-1} can also be derandomized using \texttt{COMBINE}, resulting in an $\frac{1}{\sqrt{2}-1} \approx 2.41$-competitive algorithm.  Distinct intervals would be compared based on their starting points, and the random bit would always be extracted before the arrival of an interval belonging to a different slot (compared to the very first one). \\\\
\textit{General knapsack.} Han et al. \cite{han2015randomized} study the online (with revoking) knapsack problem in the adversarial setting and give a $2$-competitive randomized algorithm that initially chooses (uniformly at random)  between two different deterministic algorithms (one that greedily maintains the best items in terms of value/size and one that is obtaining the best items in terms of value). A key point in derandomizing the algorithm is that as long as we are getting copies of the same (first) item, the two algorithms behave identically, and when a different item arrives, we will have extracted the random bit. We compare different items using the global ordering amongst distinct items (e.g. first by value and then by size). Similar to the proof for single-length interval selection, we get a $2.41$-competitive deterministic online (with revoking) algorithm for the knapsack problem  under random order arrivals.\\\\
\textit{String guessing.} In the binary  string guessing problem (see B{\"o}ckenhauer et al. \cite{bockenhauer2014string}), the algorithm has to guess an $n$-bit string, one bit at a time. While no deterministic algorithm can achieve a constant competitive ratio,  the simple randomized algorithm of  guessing the majority bit is $2$-competitive. We can derandomize this algorithm under random arrivals as follows. We start by deterministically predicting the first bit $b$ arbitrarily. Let $b'$ be the true value of the first bit (revealed after our prediction). If $b = b'$,  We keep guessing $b$ until we make a wrong prediction. At that point \texttt{COMBINE} returns a bit $r$, and we keep guessing $r$ for the remainder of the sequence. Given that \texttt{COMBINE} chooses the majority bit  with probability at least $\sqrt{2}-1$, we get a $2.41$-competitive algorithm.\\\\
\textit{Scheduling.} Albers \cite{Albers02} provides a $1.916$-competitive barely random algorithm for the makespan problem on identical machines. The algorithm initially chooses between one of two scheduling strategies, which differ in the fraction of machines they try to keep ``lightly'' and ``heavily'' loaded. It is also noteworthy that the chosen strategy keeps track of how the other strategy would have scheduled jobs so far. Because these deterministic strategies could behave differently on a stream of identical inputs, derandomizing this algorithm under random arrivals may incur an additional (beyond the bias) constant penalty to the competitive ratio.\\\\
\textit{Job throughput scheduling.} In job throughput scheduling, we usually  assume the real-time model in which jobs arrive in real time. Although it seems inconsistent to simultaneously consider real time arrivals and random order arrivals, there is a meaningful sense in which we can consider {\it real time jobs in the random order model}.  In particular, for throughput scheduling problems,  a job $J_i$ is represented as $(r_i,p_i,d_i,w_i)$ where $r_i$ is the start, arrival, or release time, $p_i$ is the processing time, $d_i$ is the deadline, and $w_i$ is the weight of the $i^{th}$ job $J_i$. The objective is to maximize the weight of a set of jobs in a feasible schedule in which  $J_i$ (if scheduled) is scheduled at some time $t_i$ such that $r_i \leq t_i, t_i + p_i \leq d_i$ and scheduled jobs do not intersect \footnote{As in interval selection, it doesn't matter if we let a job start exactly or just after another job has ended.}. One may consider a modified model of random arrivals for job scheduling as follows. The starting points of jobs are fixed, but the remaining attributes of a job (e.g. duration, deadline, weight) are randomly permuted. That is, if $J_i = (r_i,p_i,d_i,w_i)$, then the input sequence is 
$(r_1,p_{\pi(1)},d_{\pi(1)},w_{\pi(1)}), (r_2,p_{\pi(2)},d_{\pi(2)},w_{\pi(2)}), \ldots, (r_n,p_{\pi(n)},d_{\pi(n)},w_{\pi(n)})$ for a random permutation $\pi$. Note that we can equivalently represent a job as $r_i, p_i, e_i, w_i$ where $e_i = d_i - (r_i+p_i) \geq 0$  is the excess {\it slack} for scheduling the job.\\\\
Kalyanasundaram and Pruhs \cite{kalyanasundaram2003maximizing} consider a single processor setting and give a barely random algorithm with constant (although impractically large) competitiveness for the unweighted real time throughput problem. Baruah et al. \cite{baruah1994line} show that no deterministic real-time algorithm with preemption can achieve constant competitiveness. In \cite{kalyanasundaram2003maximizing},  it is shown that $OPT$ is bounded  by a linear combination of the performance of two deterministic algorithms.  We note once again that as long as the jobs arriving online are identical, both algorithms behave exactly the same. Whenever a new job arrives, the random bit will be extracted and an algorithm will be chosen. Distinct jobs can be compared based on their slack and processing time. Chrobak et al. \cite{chrobak2007online} consider the more special case of unweighted throughput problem where all jobs have the same processing time. Their $1$-bit barely random algorithm uses two copies of the same deterministic algorithm that share a lock, leading to asymmetrical behavior and achieving a competitive ratio $\frac{5}{3}$. As in Albers \cite{Albers02}, we have to simulate both algorithms regardless of which one we choose to follow. Because the two deterministic algorithms can schedule identical jobs at different times, we may incur an additional constant penalty to the competitive ratio when derandomizing under random arrivals.\\\\
\textit{Interval selection with C,D-benevolent instances.} In the context of interval selection, an instance is C-benevolent if the weight of an interval is a convex increasing function of its length, and D-benevolent if it is decreasing. Fung et al. \cite{fung2014improved} give $1$-bit barely random algorithms that are $2$-competitive for both C-benevolent and D-benevolent instances in the real-time model. These algorithms follow the logic for the case of single-length arbitrary weights, but are a bit more involved with the slots on the line being redefined adaptively throughout the execution of the algorithm. We refer the reader to \cite{fung2014improved} for a complete description of these algorithms. Although for single-length arbitrary weights we considered the conventional random ordering of all items, for the cases of C-benevolent and D-benevolent instances, we are able to derandomzize their algorithms under the \textit{real time with random order} model and get $2.41$-competitive algorithms.\\\\

\subsection{Going Beyond 1 bit of Randomness}

While our extraction processes are so far limited to one bit extraction, there are some ideas that can be used for extracting some small number of random bits depending on the application and what we can assume. In the any-order model, we know from \cite{borodin2023any} that we can extend the Fung et al. algorithm for single length instances to obtain a classify and randomly select $2k$-competitive algorithm for instances with at most $k$ different length intervals and arbitrary weights. Furthermore, $k$ does not have to be known initially, and we can randomly choose a new length whenever it first appears. In the random-order model, we can derandomize that algorithm for the special case of two-length interval selection with arbitrary weights, under the additional assumption that the input sequence consists of distinct intervals. This way we are able to utilize both Process \ref{alg:proc-1} and \ref{alg:proc-2} and get a $6$-competitive deterministic algorithm. The algorithm would use the unbiased bit from the first two intervals to decide on the length. While working on any length, we would use Process \ref{alg:proc-1} to decide on the slot type.\\\\
The above algorithm would not necessarily work in the case of three lengths, because if the first three intervals are of different length, we would not have extracted enough random bits to simulate the $\frac{1}{3}$ coin required for the third length. It would be tempting to just use the relative order of the second and third items but that bit is correlated to some extent with the first extracted bit and this introduces further difficulties. One future direction would be to study problems under the assumption that there are significant gaps in between random choices that would allow us to extract a sufficient number of bits. Another promising direction is to quantify the amount of correlation when we use consecutive input items.

\section{Conclusions}
In contrast to the optimal $2k$ competitive ratio for Algorithm \ref{alg:replace-sub} with adversarial input sequences, we have shown that the same algorithm achieves an upper bound of $2.5$ under  random order  arrivals. We have also given a lower bound of $2$ (as $n\rightarrow +\infty$) for that algorithm. We believe a matching upper bound might be possible with a more careful analysis of direct charging. It is also plausible that a deterministic algorithm can be better than $2$-competitive. A better algorithm might have additional replacement rules, in particular for partial conflicts. We also want to improve the $\frac{12}{11}$ lower bound for random arrivals as we believe  that this is not the optimal bound. The reason to present this easy lower bound is to demonstrate the provable difference between real time and random order arrivals.\\\\
 Our study has thus far only considered deterministic algorithms and an obvious question is to consider randomized algorithms for interval selection with revoking in the random order model. More generally, we ask if it is possible to obtain a deterministic or randomized constant competitive ratio for interval selection and arbitrary weights with random order arrivals and revoking. We conjecture that this is not possible.  One possibility is to try to adapt the Canetti and Irani \cite{chrobak2006note} randomized lower bound with revoking for interval selection with arbitrary weights. Alternatively, for what other classes of weight functions (beyond C-Benevolent and D-Benevolent) can we achieve a constant ratio?\\\\
We studied three processes for extracting random bits from uniformly random arrivals, and showed direct applications in derandomizing algorithms for a variety of problems. So far, our results are mainly limited to extracting a single random bit. We considered interval selection (with arbitrary weights) when there are at most two interval lengths assuming that all intervals in the input sequence are distinct.  With that assumption, we  were able to extract two bits so as to achieve a constant competitive ratio. Can we obtain a constant ratio without this strong assumption? Furthermore, can we achieve a constant ratio by extracting random bits  (in terms of $k$)  for when there are at most $k > 2$ different interval lengths when assuming that all input items are different? Are there other applications which only gradually need random bits and not all the random  bits initially as is usually done in most classify and randomly select algorithms. Although the Fung et al. algorithm for single length instances can be viewed as a classify and randomly select algorithm, we know from \cite{borodin2023any} that random bits do not have to be initially known for instances with a limited number of interval lengths.\\\\
We conclude that we have a very limited understanding of the power of the random arrival model with revoking.  We repeat a question that motivated our interest in online randomness extraction. Namely, is there a (natural or contrived) problem for which we can provably show a (significantly)  better randomized competitive algorithm with adversarial order than what is achievable by deterministic algorithms with random order arrivals?

\printbibliography
\appendix
\section{Dealing with the cases of $|S|<3$}\label{app:A}
\subsection{In Lemma \ref{lemm:k_2_main}.}
Consider the case of $|S| = 1$, with $S=\{S_{1}\}$. As per the analysis of Lemma \ref{lemm:k_2_main}, we get that $\E[TC(S_{1})] \leq 1$. If $S_{1}$ is an optimal interval, it means that it is only directly charged when it is accepted, and we have that $DC(S_{1}) = 1$. Therefore, $\E[\Phi(S_{1})]\leq 2$. If $S_{1}$ is not an optimal interval (because it conflicts with the left and/or right optimal interval), we have that $\Phi(S_{1}) \leq 2$ just from the fact that there exist two optimal intervals in total.\\\\
We now consider the case of $|S| = 2$, with $S=\{S_{1},S_{2}\}$. We have that for $I\in \{S_{1},S_{2}\}$, $\E[TC(I)]\leq \frac{2}{3}$. If both $S_{1},S_{2}$ are optimal intervals, we have that $DC(S_{1}) = DC(S_{2}) = 1$, and therefore for $I\in \{S_{1},S_{2}\}$, $\E[\Phi(I)]\leq \frac{2}{3} + 1 < 2.5$. If $S_{1},S_{2}$ are both non-optimal intervals, we have that $\Phi(S_{i}) \leq 2$ for all $i$, because there exist at most two optimal intervals in total. Consider now the case of $S_{1}$ being an optimal interval, and $S_{2}$ being a non-optimal interval. We consider two subsequent cases:\\\\
\textit{Case 1:} $S_{1}$ and $S_{2}$ don't overlap. In this case we have that $DC(S_{1}) = 1$, and $DC(S_{2})\leq 1$, since $S_{2}$ must be in conflict with either the left or the right optimal interval. Therefore we get $\E[\Phi(I)] < 2.5$ for $I\in \{S_{1},S_{2}\}$.\\\\
\textit{Case 2:} $S_{1}$ and $S_{2}$ overlap. In this case, we have that $DC(S_{1}) = 1$, and $1 \leq DC(S_{2}) \leq 2$. Let $I\in \{S_{1},S_{2}\}$ be the interval that makes it into the final solution. We have that $\E[DC(I)]\leq \frac{1}{2}\cdot 1 + \frac{1}{2} \cdot 2 = 1.5$, and $\E[\Phi(I)]\leq \frac{2}{3} + 1.5 < 2.5$.\\\\
In conclusion, Lemma \ref{lemm:k_2_main} holds even in the case of $|S| < 3$.
\subsection{In Lemma \ref{lemm:big_k_main}.}
The induction argument in Lemma \ref{lemm:big_k_main} goes through with $|S'_d| = 1$ and $|S'_d| = 2$, giving us a bound on the expected amount of transfer charge of $1$ and $\frac{2}{3}$ respectively. As in A.1, in the case of $|S'_d| = 1$ with $S_1$ being an optimal interval, we have that $DC(S_1) = 1$. If $S_1$ is not an optimal interval and conflicts with one optimal interval, we have that $DC(S_1)\leq 1$. If $S_1$ conflicts with two optimal intervals, we have that $DC(S_1)= 2$ with probability at most $\frac{1}{3}$. Therefore $\E[DC(S_1)]\leq \frac{2}{3} + 2\cdot \frac{1}{3} = \frac{4}{3}$. In all cases, we have that $\E[TC(S_1)] + \E[DC(S_1)] \leq 2.5$.\\\\
In the case of $S'_d = \{S_1, S_2\}$ with both intervals being optimal, as in A.1, we get that $DC(S_1) = DC(S_2) = 1$. If neither interval is optimal, similar to the argument in the case of $|S'_d| = 1$, we have that $DC(S_i) = 2$ with probability at most $\frac{1}{2}$, and $\E[DC(S_i)]\leq \frac{1}{2} + 2 \cdot \frac{1}{2}$. Therefore, $\E[\Phi(S_i)] \leq 2.5$. Finally, the case of one of the two intervals being optimal is handled like in A.1.
\section{A lower bound under random arrivals.}\label{app:B}
We will show a lower bound of $\frac{12}{11}$ for all deterministic algorithms with revoking in the random-order model. This is in contrast to the real-time model with revoking, where $1$-competitiveness is attainable. Let $LB_1$ be a three interval instance as shown in figure \ref{fig:LB-1}.

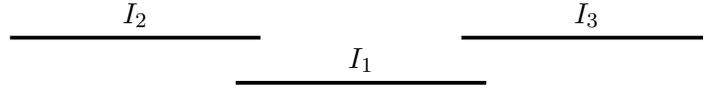
\begin{figure}[H]
	\centering
	
	\begin{tikzpicture}[scale=0.6]

	\node[draw=none] (I2a) at (-18,1) {$ $};
	\node[draw=none] (I2b) at (-12,1) {$ $};
	\node at (-15,1.5) {$I_{2}$};
	\draw[line width=0.5mm] (I2a) -- (I2b);

    \node[draw=none] (I3a) at (-13,0) {$ $};
	\node[draw=none] (I3b) at (-7,0) {$ $};
 \node at (-10,0.5) {$I_{1}$};
	\draw[line width=0.5mm] (I3a) -- (I3b);

 \node[draw=none] (I4a) at (-8,1) {$ $};
	\node[draw=none] (I4b) at (-2,1) {$ $};
	\node at (-5,1.5) {$I_{3}$};
	\draw[line width=0.5mm] (I4a) -- (I4b);

	\end{tikzpicture} 
	\caption{Instance $LB_1$.}\label{fig:LB-1}
\end{figure}
First, notice that because the algorithm has no knowledge of the size of the input, it must act greedily on the first interval to arrive. If that weren't the case, we could introduce a one-interval instance where the competitiveness of the algorithm would be unbounded. Consider now the behaviour of the algorithm if $I_1$ was to arrive first.\\\\
\textit{Case 1:} The algorithm will not replace $I_1$ with either interval that might arrive second. In this case, we know that with probability at least $\frac{1}{3}$, the algorithm will have one interval in its solution.\\\\
\textit{Case 2:} There is at least one interval in $\{I_2,I_3\}$ such that if it is the second interval to arrive, it will replace $I_1$. Let $I_3$ be such an interval. We  can then use instance $LB_2$ (fig. \ref{fig:LB-2}), with intervals $I_1$ and $I_3$ being the same as in $LB_1$. In this case we know that with probability at least $\frac{1}{6}$, the algorithm will have one interval in its solution.
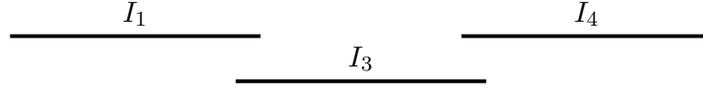
\begin{figure}[H]
	\centering
	
	\begin{tikzpicture}[scale=0.6]

	\node[draw=none] (I2a) at (-18,1) {$ $};
	\node[draw=none] (I2b) at (-12,1) {$ $};
	\node at (-15,1.5) {$I_{1}$};
	\draw[line width=0.5mm] (I2a) -- (I2b);

    \node[draw=none] (I3a) at (-13,0) {$ $};
	\node[draw=none] (I3b) at (-7,0) {$ $};
 \node at (-10,0.5) {$I_{3}$};
	\draw[line width=0.5mm] (I3a) -- (I3b);

 \node[draw=none] (I4a) at (-8,1) {$ $};
	\node[draw=none] (I4b) at (-2,1) {$ $};
	\node at (-5,1.5) {$I_{4}$};
	\draw[line width=0.5mm] (I4a) -- (I4b);

	\end{tikzpicture} 
	\caption{Instance $LB_2$.}\label{fig:LB-2}
\end{figure}
In conclusion, there is always an instance with $OPT=2$, and $\E[ALG]\leq \frac{1}{6}+2\cdot\frac{5}{6}=\frac{11}{6}$, and therefore the competitive ratio is at least $\frac{12}{11}$. Even under the assumption that the algorithm knows the size of the input, the same bound holds. If $I_1$ was not taken upon arrival, using instance $LB_2$ we know that the algorithm will only have one interval in its solution with probability at least $\frac{1}{3}$.

\end{document}